\documentclass[AMA,STIX1COL]{WileyNJD-v2}

\articletype{Research Article}

\received{26 April 2016}
\revised{6 June 2016}
\accepted{6 June 2016}

\usepackage{caption}
\usepackage{subcaption}

\begin{document}

\title{On the Comparison between the Reliability of Units Produced by Different Production Lines}

\author[1]{Rashad M. EL-Sagheer}

\author[1]{M. A. W. Mahmoud}

\author[2]{M. M. M. Mansour}

\author[2]{Mohamed S. Aboshady*}

\authormark{RASHAD EL-SAGHEER \textsc{et al}}

\address[1]{\orgdiv{Department of Mathematics}, \orgname{Faculty of Science, Al-Azhar University}, \orgaddress{\state{Cairo}, \country{Egypt}}}

\address[2]{\orgdiv{Department of Basic Science}, \orgname{Faculty of Engineering, The British University in Egypt}, \orgaddress{\state{Cairo}, \country{Egypt}}}

\corres{*Mohamed Saied Aboshady\\ \email{Mohamed.Aboshady@bue.edu.eg}}


\abstract[Abstract]{The paper discusses how to evaluate the reliability of units produced by different production lines. The procedure is based on selecting independent random samples of units produced by different production lines and then evaluating reliability functions for each group of units. The comparison between these reliability functions at a given time allows manufacturing experts to evaluate the effectiveness of production lines. A statistical methodology has been taken based on the assumption that the lifetime of units produced by each product line has a Weibull Gamma distribution. Then, real-world data is used to illustrate the study's contribution to reliability theory applications.}

\keywords{Joint progressive Type-II censoring, Bias Corrected Confidence Interval (Boot-BC), Accelerated Bias Corrected Confidence Interval (Boot-BCa), Comparing the reliability between $k$ production lines, MCMC technique}


\maketitle


\section{Introduction}\label{sec1}

The progressive Type-II censoring acquires its importance from giving the experimenter the flexibility to remove some units from the test at non-terminal time points. This feature leads to reduce the cost of carrying out such experiments and accelerates the lifetime of the test. Life testing aims to obtain information about the population parameters through testing some units, where knowing this information or data help statisticians to estimate, for instance, the reliability function after fitting these data to the appropriate population. Several authors have studied the estimation problems based on progressive Type-II censored samples, for example, Wu and Chang \cite{1}, Abdel-Hamid and AL-Hussaini \cite{2}, Ahmed \cite{3}, El-Sagheer and Ahsanullah \cite{4} and Abdel-Hamid \cite{5}. The joint censoring scheme has been previously developed in the literature. For example, Balakrishnan and Rasouli \cite{6} discussed exact inference for two exponential populations when Type-II censoring is implemented on the two samples in a combined manner. Balakrishnan and Feng \cite{7} generalized their work to the case of $k$ samples situation. Recently, statisticians turn their attention to a new development of progressive Type-II censoring which is called joint progressive Type-II censoring (JPROG-II-C). Balakrishnan and Rasouli \cite{8} subsequently extended their work to the case of two exponential populations when JPROG-II-C is implemented on the two samples. Doostparast et al. \cite{9} estimated the parameters of two Weibull populations based on JPROG-II-C samples using the Bayesian estimation under LINEX loss function. Balakrishnan et al. \cite{10} generalized the work of Balakrishnan and Rasouli \cite{8} to $k$ exponential populations. The JPROG-II-C samples can be observed mostly in experimental situations. Assume that units are being manufactured by $k$ different lines within the same facility. From each production line $h$, a random sample of size $n_{h},~h=1,2,...,k,$ is selected where the samples are independent. Now, these independent samples are simultaneously put to life testing. The test may be expensive or time-consuming, so the experimenter will proceed to apply the progressive Type-II strategy. Accordingly, the experimenter will terminate the experiment when the $r$th failure occurs. In this paper, $k$ independent samples from WG populations are put under life testing based on JPROG-II-C. The remainder of this paper is organized as follows: Section 2 provides a description of WG distribution (WGD) and the test steps. In Section 3, the maximum likelihood estimates (MLEs) of the parameters under consideration are obtained in addition to the corresponding ACIs and four types of bootstrap confidence intervals are also computed. The comparison between the reliability of $k$ production lines is also developed in Section 3. Section 4 is devoted to the Bayesian approach that uses the famed MCMC technique. An illustrative example is presented to explain the theoretical results at $k=2$ in Section 5. In addition to computing the expected values of the number of failures before the test procedure, the exact values of the number of failures are also observed after the test procedure. A simulation study is performed to compare the expected values with the exact values in Section 6 at $k=2$ in addition to applying on a real data example. Finally, the conclusion is placed in Section 7.

\section{Model description}\label{sec2}

This section contains two subsections. The first is devoted to giving a brief description of WGD. The likelihood function under $k$ populations is presented based on JPROG-II-C in the second subsection.

\subsection{Weibull Gamma Distribution}

The WG is a suitable distribution for the phenomenon of loss of signals in telecommunications, which is called fading, when multi-path is superimposed on shadowing, see Bithas \cite{11}. The WGD had been introduced by Nadarajah and kotz \cite{12}. A random variable $X$ is said to have WGD with scale parameter $\alpha$ if its probability density function (PDF) is given by:
\begin{equation}
f\left( x;\alpha ,\theta ,\beta \right) =\frac{\theta \beta }{\alpha }\left(\frac{x}{\alpha }\right) ^{\theta -1}\left( 1+\left( \frac{x}{\alpha}\right) ^{\theta }\right) ^{-\left( \beta +1\right) }\;,x>0;~\alpha ,\theta,\beta >0,  \label{2.1.1}
\end{equation}
and the reliability function at any time $t$ is 

\begin{equation}
R(t)=P(X>t)=\left( 1+\left( \frac{t}{\alpha }\right) ^{\theta }\right)^{-\beta },t>0.  \label{2.1.2}
\end{equation}

For more details about WGD and its properties see Molenberghs and Verbeke \cite{13} and Mahmoud et al. \cite{14}. The estimation of WGD parameters is discussed by EL-Sagheer \cite{15} based on progressive Type-II censored samples.

\subsection{Implementation of JPROG-II-C}

Suppose that $X_{h1},X_{h2},...,X_{hn_{h}}$ are the lifetimes of $n_{h}$ units from line $A_{h}$, and assume that they are independent and identically distributed (iid) variables from a population with cumulative distribution function (CDF) $F_{h}(x)$ and PDF $f_{h}(x),h=1,2,...,k$. Let $N=\sum_{h=1}^{k}n_{h}$ denotes the total sample size, $r$ denotes the total number of the observed failures and $\zeta _{1}\leq \zeta _{2}\leq ...\leq \zeta_{N}$ represent the order statistics of the $N$ random variables $\{X_{h_{i}1},X_{h_{i}2},...,X_{h_{i}n_{i}}\}.$ When JPROG-II-C is implemented on $k$ samples, the observable data will consist of $\left(\delta,\zeta\right)$, where $\zeta =(\zeta _{1}\leq\zeta _{2}\leq ...\leq \zeta _{r}),~\zeta _{i}\in \{X_{h_{i}1},X_{h_{i}2},...,X_{h_{i}r}\}$ for, $h_{i}=1,2,...,k,$ $i=1,2,...,r$, and depending on $\mathbf{(}h_{1}$ $\mathbf{\leq }$ $h_{2}$ $\mathbf{\leq ...\leq }$ $h_{r}\mathbf{)}$, $\delta$ can be defined as: $\delta =\left( \delta _{1}(h),\delta_{2}(h),...,\delta _{r}(h)\right) ,$ where 
\begin{equation}
\delta _{i}(h)=\left\{ 
\begin{array}{l}
1, \\ 
0,%
\end{array}%
\begin{array}{l}
\text{if }h=h_{i} \\ 
\multicolumn{1}{c}{\text{elsewhere,}}%
\end{array}%
\right.
\end{equation}
with $\sum_{h=1}^{k}\delta _{i}(h)=1$, for $i=1,2,...,r.$ After the first failure occurrence at time $\zeta _{1},$ $R_{1}$ units are randomly removed from the remaining $N-1$ surviving units. At time $\zeta _{2}$, the second failure happened and $R_{2}$ units are randomly removed from the remaining $N-R_{1}-2$ surviving units. Finally, at time $\zeta _{r}$, $R_{r}=N-r-\sum_{i=1}^{r-1}R_{i}$ surviving units are removed from the test when the $r$th failure happened. A detailed plan for PROG-II-C is discussed by Balakrishnan and Aggarwala \cite{16}. Let $s_{i}(h)$ represent the number of the removed units at the occurrence time of the $i$th failure from sample $h$. Hence $R_{i}=\sum_{h=1}^{k}s_{i}(h),$ which is previously determined by the experimenter, for $i=1,2,...,r$. According to the JPROG-II-C, it is possible to assume $r=\sum_{h=1}^{k}M_{r}(h)$ where $M_{r}(h)=\sum_{i=1}^{r}\delta_{i}(h)$ represents the number of failures occurred in sample $h$.

\section{Maximum Likelihood Estimation}\label{sec3}

The log-likelihood functions are the basis for deriving estimators of parameters, given data. MLEs are characterized by different advantages such as asymptotically normally distributed, asymptotically minimum variance, asymptotically unbiased, and satisfy the invariant property, see Azzalini \cite{17} and Royall \cite{18} for more information on likelihood theory. Balakrishnan et al. \cite{10} introduced the likelihood function in a general form for $k$ samples drawn from any PDF under JPROG-II-C as:
\begin{equation}
L(\alpha _{1,2,...,k},\theta _{1,2,...,k},\beta _{1,2,...,k},\delta,s,\zeta)=c_{r}\prod_{i=1}^{r}\prod_{h=1}^{k}\left(f_{h}\left( \zeta _{i}\right) \right) ^{\delta_{i}(h)}\prod_{i=1}^{r}\prod_{h=1}^{k}\left( \bar{F}_{h}\left( \zeta_{i}\right) \right) ^{s_{i}(h)},  \label{3.1}
\end{equation}
where $\mathbf{s=}\left\{ (s_{1}\left( h\right) ,s_{2}\left( h\right),...,s_{r}\left( h\right) )\right\} _{h=1}^{k},$ $\bar{F}_{h}\left( \zeta_{i}\right) =1-F_{h}\left( \zeta _{i}\right) $ and $c_{r}=D_{1}D_{2}$, with
\begin{equation*}
D_{1}=\prod_{i=1}^{r}\left[ \sum_{h=1}^{k}\left( n_{h}-M_{i-1}\left(
h\right) -\sum_{j=1}^{i-1}s_{j}(h)\right) \delta _{i}(h)\right] ,
\end{equation*}
and 
\begin{equation*}
D_{2}=\prod_{i=1}^{r-1}\left\{ \frac{\prod_{h=1}^{k}\binom{n_{h}-M_{i}\left(
h\right) -\sum_{j=1}^{i-1}s_{j}(h)}{s_{j}(h)}}{\binom{N-i-%
\sum_{j=1}^{i-1}R_{j}}{R_{i}}}\right\}.
\end{equation*}
If the removal of units occur only at the terminal point, then we will obtain the joint Type-II censoring scheme as a special case of JPROG-II-C, i.e. $s_{i}(h)=0$, therefore $R_{i}=0$ when $i=1,2,...,r-1.$ and $s_{r}(h)=n_{h}-M_{r}(h),$ for all $h=1,2,...,k,$ hence $R_{r}=\sum_{h=1}^{k}s_{r}(h)=N-r.$ For $\bar{F}_{h}\left( \zeta _{i}\right)=\left( 1+\left( \frac{\zeta _{i}}{\alpha _{h}}\right) ^{\theta _{h}}\right)^{-\beta _{h}}$ and $f_{h}\left( \zeta _{i}\right) =\frac{\theta _{h}\beta_{h}}{\alpha _{h}}\left( \frac{\zeta _{i}}{\alpha _{h}}\right) ^{\theta_{h}-1}\left( 1+\left( \frac{\zeta _{i}}{\alpha _{h}}\right) ^{\theta_{h}}\right) ^{-\left( \beta _{h}+1\right) },~i=1,2,...,r,~h=1,2,...,k$, representing the survival function and PDF for the test units' respectively, equation \eqref{3.1} can be written as follows:
\begin{eqnarray}
L(\alpha _{1,2,...,k},\theta _{1,2,...,k},\beta _{1,2,...,k},\delta,s,\zeta)=c_{r}\prod_{h=1}^{k}\left( \frac{\theta_{h}\beta _{h}}{\alpha _{h}}\right) ^{M_{r}(h)} \times \exp \left\{ \left( \theta _{h}-1\right) \left[ \sum_{i=1}^{r}\delta _{i}(h)\log \zeta _{i}-M_{r}(h)\log \alpha _{h}\right] \right\} \nonumber \\ \times \exp \left\{ -\sum_{i=1}^{r}\left[ \left( \beta _{h}\left( \delta_{i}(h)+s_{i}(h)\right) +\delta _{i}(h)\right) \log \left( 1+\left( \frac{\zeta _{i}}{\alpha _{h}}\right) ^{\theta _{h}}\right) \right] \right\} .
\label{3.2}
\end{eqnarray}

The log-likelihood function may then be written as:
\begin{eqnarray*}
\log L=\log c_{r}+\sum_{h=1}^{k}M_{r}(h)\left( \log \theta _{h}+\log \beta_{h}-\theta _{h}\log \alpha _{h}\right) \\ + \sum_{h=1}^{k}\sum_{i=1}^{r}\left[ \left( \theta _{h}-1\right) \delta_{i}(h)\log \zeta _{i}-\left( \beta _{h}\left( \delta_{i}(h)+s_{i}(h)\right) +\delta _{i}(h)\right) \log \left( 1+\left( \frac{\zeta _{i}}{\alpha _{h}}\right) ^{\theta _{h}}\right) \right] ,
\end{eqnarray*}

and thus we have the likelihood equations for $\alpha _{h},\theta _{h}$ and $\beta _{h},~h=1,2,...,k$, respectively, as
\begin{equation}
\frac{M_{r}(h)\theta _{h}}{\alpha _{h}}-\frac{\theta _{h}}{\alpha_{h}^{\theta _{h}+1}}\sum_{i=1}^{r}\frac{\zeta _{i}^{\theta _{h}}\left(\beta _{h}\left( \delta _{i}(h)+s_{i}(h)\right) +\delta _{i}(h)\right) }{1+\left( \frac{\zeta _{i}}{\alpha _{h}}\right) ^{\theta _{h}}}=0,
\label{3.3}
\end{equation}

\begin{equation}
\left( 1-\theta _{h}\log \alpha _{h}\right) \frac{M_{r}(h)}{\theta _{h}}+\sum_{i=1}^{r}\left[ \delta _{i}(h)\log \zeta _{i}-\frac{1}{\alpha_{h}^{\theta _{h}}}\left( \frac{\zeta _{i}^{\theta _{h}}\left( \beta_{h}\left( \delta _{i}(h)+s_{i}(h)\right) +\delta _{i}(h)\right) \log \left( \frac{\zeta _{i}}{\alpha _{h}}\right) }{1+\left( \frac{\zeta _{i}}{\alpha_{h}}\right) ^{\theta _{h}}}\right) \right] =0,  \label{3.4}
\end{equation}
and
\begin{equation}
\frac{M_{r}(h)}{\beta _{h}}-\sum_{i=1}^{r}\left( \delta
_{i}(h)+s_{i}(h)\right) \log \left( 1+\left( \frac{\zeta _{i}}{\alpha _{h}}%
\right) ^{\theta _{h}}\right) =0.  \label{3.5}
\end{equation}

Since the above Equations \eqref{3.3}, \eqref{3.4} and \eqref{3.5} are nonlinear simultaneous equations in $3k$ unknown variables $\alpha _{h}$, $\theta _{h}$ and $\beta _{h}$ for $h=1,2,...,k$, it is obvious that an exact solution is not easy to obtain. Therefore, a numerical method such as Newton Raphson can be used to find approximate solution.The algorithm has been implemented using the following steps:

\begin{enumerate}
\item[(1)] Use the method of moments or some other proper estimates of the parameters as initial points for iteration, denote the initials as $(\alpha_{h}\left( 0\right) ,\theta _{h}(0),\beta _{h}(0))$ for the parameters $(\alpha _{h},\theta _{h},\beta _{h})$ and set $j=0.$

\item[(2)] Calculate $\left( \frac{\partial \log L}{\partial \alpha _{h}},%
\frac{\partial \log L}{\partial \theta _{h}},\frac{\partial \log L}{\partial
\beta _{h}}\right) _{\left( \alpha _{h}\left( j\right) ,\theta _{h}\left(
j\right) ,\beta _{h}(j)\right) }$ and the observed asymptotic Fisher
information matrix $I$ $^{-1}(\alpha _{h},\theta _{h},\beta _{h})$, given in
the next section.

\item[(3)] Update $(\alpha _{h},\theta _{h},\beta _{h})$ as
\begin{equation*}
\left( \alpha _{h}(j+1),\theta _{h}\left( j+1\right) ,\beta _{h}(j+1)\right)=\left( \alpha _{h}\left( j\right) ,\theta _{h}\left( j\right) ,\beta_{h}(j)\right)-\left( \frac{\partial \log L}{\partial \alpha _{h}},\frac{\partial \log L}{\partial \theta _{h}},\frac{\partial \log L}{\partial \beta_{h}}\right) _{\left( \alpha _{h}\left( j\right) ,\theta _{h}\left( j\right),\beta _{h}(j)\right) } \times I^{-1}(\alpha _{h},\theta _{h},\beta _{h}).
\end{equation*}

\item[(4)] Put $j=j+1$, and then return to step $2.$

\item[(5)] Continue the consecutive steps until
\begin{equation*}
\left\vert \left( \alpha _{h}(j+1),\theta _{h}\left( j+1\right) ,\beta
_{h}(j+1)\right) -\left( \alpha _{h}\left( j\right) ,\theta _{h}\left(
j\right) ,\beta _{h}(j)\right) \right\vert \leq \epsilon \rightarrow 0.%
\newline
\end{equation*}
The final estimates of $(\alpha _{h},\theta _{h},\beta _{h})$ are the MLEs
of the parameters, denoted as$\left( \hat{\alpha}_{h},\hat{\theta}_{h},\hat{%
\beta}_{h}\right) $ for $h=1,2,...,k.$
\end{enumerate}

In the next two subsections, ACIs and four different bootstrap confidence intervals for the parameters $(\alpha _{h},\theta _{h},\beta _{h})$ are deduced based on the MLEs.

\subsection{Approximate confidence intervals}

The $\left( 1-\vartheta \right) 100\%$ \ ACIs for the parameters $\alpha
_{h},\theta _{h}$ and $\beta _{h}$ $,h=1,2,...,k$, can be written as: 
\begin{equation*}
\begin{array}{c}
(\hat{\alpha}_{hL},\hat{\alpha}_{hU})=\hat{\alpha}_{h}\pm z_{1-\frac{\zeta }{%
2}}\sqrt{var(\hat{\alpha}_{h})} \\ 
(\hat{\theta}_{hL},\hat{\theta}_{hU})=\hat{\theta}_{h}\pm z_{1-\frac{\zeta }{%
2}}\sqrt{var(\hat{\theta}_{h})} \\ 
(\hat{\beta}_{hL},\hat{\beta}_{hU})=\hat{\beta}_{h}\pm z_{1-\frac{\zeta }{2}}%
\sqrt{var(\hat{\beta}_{h})}%
\end{array},
\end{equation*}
where $z_{1-\frac{\vartheta }{2}}$is the percentile of the standard normal distribution with left-tail probability $1-\frac{\vartheta }{2}$, and $var(\hat{\alpha}_{h}),var(\hat{\theta}_{h})$ and $var(\hat{\beta}_{h})$ represent the asymptotic variances of the maximum likelihood estimates which can be calculated using the inverse of the Fisher information matrix. Let $I(\Omega _{1_{h}},\Omega _{2_{h}},\Omega _{3_{h}})$ $=$ $I(\Omega_{1_{1}},...,\Omega _{1_{k}},\Omega _{2_{1}},...,\Omega _{2_{k}},\Omega_{3_{1}},...,\Omega _{3k})$ denotes the asymptotic Fisher information matrix of the parameters $\Omega _{1_{h}}=\alpha _{h},\;\Omega _{2_{h}}=\theta_{h}\;$and $\Omega _{3_{h}}=\beta _{h},h=1,2,...,k,$
where 
\begin{equation*}
I(\Omega _{1_{h}},\Omega _{2_{h}},\Omega _{3_{h}})=-\left( \frac{\partial
^{2}\log L}{\partial \Omega _{i_{h}}\partial \Omega _{j_{h}}}\right)
,i,j=1,2,3\text{, for each }h.
\end{equation*}

The asymptotic variance--covariance matrix for the maximum likelihood estimates can then be calculated as follows:
\begin{equation}
I^{-1}=\left[ -\left( \frac{\partial ^{2}\log L}{\partial \Omega
_{i_{h}}\partial \Omega _{j_{h}}}\right) \right] _{\downarrow \left( \hat{%
\Omega}_{1_{h}},\hat{\Omega}_{2_{h}},\hat{\Omega}_{3_{h}}\right) }^{-1},
\label{3.7}
\end{equation}%
for more details see Cohen \cite{19}.

\subsection{Bootstrap confidence intervals}

The Bootstrap confidence intervals are proposed based on the parametric bootstrap methods where the parametric model for the data is known, $f\left(\zeta;.\right)$, up to the unknown parameters $\left( \alpha _{h},\theta _{h},\beta _{h}\right),$ such that the bootstrap data are sampled from $f\left(\zeta;\hat{\alpha}_{h},\hat{\theta}_{h},\hat{\beta}_{h}\right)$, where $\left(\hat{\alpha}_{h},\hat{\theta}_{h},\hat{\beta}_{h}\right) $ are the MLEs from the original data. A lot of papers dealt only with percentile bootstrap method (Boot-p) based on the idea of Efron \cite{20} and bootstrap-t method (Boot-t) based on the idea of Hall \cite{21}, such as Soliman et al. \cite{22}, El-Sagheer \cite{23} and others. In this article, two additional types of Bootstrap CIs, Boot-BC and Boot-BCa, based on the idea of DiCiccio and Efron \cite{24}, are discussed. The following algorithm is followed to obtain bootstrap samples for the four methods:

\begin{enumerate}
\item[(1)] Based on the original JPROG-II-C sample, $\zeta _{1}$ $\mathbf{%
\leq }$ $\zeta _{2}\mathbf{\leq ...\leq }$ $\zeta _{r},$ compute $\hat{\alpha}%
_{h},\hat{\theta}_{h}$ and $\hat{\beta}_{h}$ for $h=1,2,...,k.$

\item[(2)] Use $\hat{\alpha}_{h},\hat{\theta}_{h}$ and $\hat{\beta}_{h}$ to generate a bootstrap sample $\zeta^{\ast }$ with the same
values of $R_{i},\;i=1,2,...,r\;$using algorithm presented in Balakrishnan
and Sandhu \cite{25}.

\item[(3)] As in step (1) and based on $\zeta^{\ast }$ compute the
bootstrap sample estimates for $\hat{\alpha}_{h},\hat{\theta}_{h}$ and $\hat{%
\beta}_{h}$, say $\hat{\alpha}_{h}^{\ast },\hat{\theta}_{h}^{\ast }$ and $\hat{%
\beta}_{h}^{\ast }.$

\item[(4)] Repeat the previous steps (2) and (3) $B$ times and arrange all $%
\hat{\alpha}_{h}^{\ast },\hat{\theta}_{h}^{\ast }$ and $\hat{\beta}%
_{h}^{\ast }\;$in ascending order to obtain the bootstrap sample $\left( 
\hat{\Omega}_{j_{h}}^{\ast \lbrack 1]},\hat{\Omega}_{j_{h}}^{\ast \lbrack
2]},...,\hat{\Omega}_{j_{h}}^{\ast \lbrack B]}\right) ,\;j=1,2,3,$ where $%
\hat{\Omega}_{1_{h}}^{\ast }=\hat{\alpha}_{h}^{\ast },\;\hat{\Omega}%
_{2_{h}}^{\ast }=\hat{\theta}_{h}^{\ast },\;\;\hat{\Omega}_{3_{h}}^{\ast }=%
\hat{\beta}_{h}^{\ast }.$
\end{enumerate}

\subsubsection{Bootstrap-p confidence interval}

Let $\Phi (z)=P(\hat{\Omega}_{j_{h}}^{\ast }\leq z)\;$be the cumulative
distribution function of $\hat{\Omega}_{j_{h}}^{\ast }.\;$Define $\hat{\Omega%
}_{j_{h}Boot}^{\ast }=\Phi ^{-1}(z)\;$for a given $z$. The approximate
bootstrap-p $100(1-\vartheta )\%\;$confidence interval of $\hat{\Omega}%
_{k}^{\ast }\;$is given by:
\begin{equation*}
\;%
\begin{bmatrix}
\;\hat{\Omega}_{j_{h}Boot}^{\ast }(\frac{\vartheta }{2}) & , & \hat{\Omega}%
_{j_{h}Boot}^{\ast }(1-\frac{\vartheta }{2})\;%
\end{bmatrix}%
.\;
\end{equation*}

\subsubsection{Bootstrap-t confidence interval}

Consider the order statistics $\mu _{j_{h}}^{\ast \lbrack 1]}<\;\mu
_{j_{h}}^{\ast \lbrack 2]}<\;...\;<\;\mu _{j_{h}}^{\ast \lbrack B]}\;$where%
\begin{equation*}
\mu _{j_{h}}^{\ast \lbrack p]}=\frac{\sqrt{B}(\hat{\Omega}_{j_{h}}^{\ast
\lbrack p]}-\hat{\Omega}_{j_{h}})}{\sqrt{Var\left( \hat{\Omega}%
_{j_{h}}^{\ast \lbrack p]}\right) }},\;p=1,2,...,B;\;j=1,2,3;\text{for }%
h=1,2,...,k.
\end{equation*}%
where $\hat{\Omega}_{1_{h}}=\hat{\alpha}_{h},\;\hat{\Omega}_{2_{h}}=\hat{%
\theta}_{h}\;$and $\hat{\Omega}_{3_{h}}=\hat{\beta}_{h}\;$while $Var\left( 
\hat{\Omega}_{j_{h}}^{\ast \lbrack p]}\right) \;$is obtained using the
inverse of the Fisher information matrix as done before in \eqref{3.7}. Let $%
W\left( z\right) =P\left( \mu _{j_{h}}^{\ast }<z\right) ,~j=1,2,3$ be the
cumulative distribution function of $\mu _{j_{h}}^{\ast }.$\newline
For a given $z,$ define%
\begin{equation*}
\hat{\Omega}_{j_{h}Boot\;-t}^{\ast }=\hat{\Omega}_{j_{h}}+B^{\frac{-1}{2}}%
\sqrt{Var\left( \hat{\Omega}_{j_{h}}^{\ast }\right) }W^{-1}\left( z\right) .
\end{equation*}%
Thus, the approximate bootstrap-t $100(1-\vartheta )\%\;$confidence interval
of $\hat{\Omega}_{j_{h}}^{\ast }\;$is given by:

\begin{equation*}
\;%
\begin{bmatrix}
\hat{\Omega}_{j_{h}Boot\;-t}^{\ast }(\frac{\vartheta }{2}) & , & \hat{\Omega}%
_{j_{h}Boot\;-t}^{\ast }(1-\frac{\vartheta }{2})\;%
\end{bmatrix}%
.\;
\end{equation*}

\subsubsection{Bootstrap bias corrected confidence interval}

Let $\Phi (z)=\vartheta \;$be the standard normal cumulative distribution
function, with $z_{\vartheta }=\Phi ^{-1}(\vartheta ).$ Define the
bias-correction constant $z_{\circ }\;$from the following probability $P(%
\hat{\Omega}_{j_{h}}^{\ast }\leq \hat{\Omega}_{j_{h}})=G(z_{\circ
}),~j=1,2,3,\;$ where $G(.)\;$is the cumulative distribution function of the
bootstrap distribution and 
\begin{equation*}
P(\hat{\Omega}_{j_{h}}^{\ast }\leq \hat{\Omega}_{j_{h}})=\frac{\#%
\begin{Bmatrix}
\hat{\Omega}_{j_{h}}^{\ast \lbrack p]}\;<\;\hat{\Omega}_{j_{h}}%
\end{Bmatrix}%
}{B},\;p=1,2,...,B;\;j=1,2,3;\text{for }h=1,2,...,k.
\end{equation*}%
thus 
\begin{equation}
z_{\circ }=\Phi ^{-1}%
\begin{pmatrix}
\frac{\#%
\begin{Bmatrix}
\hat{\Omega}_{j_{h}}^{\ast \lbrack p]}\;<\;\hat{\Omega}_{j_{h}}%
\end{Bmatrix}%
}{B}%
\end{pmatrix}%
,\;p=1,2,...,B;\;j=1,2,3;\text{for }h=1,2,...,k.  \label{3.2.0}
\end{equation}%
For a given $\vartheta ,$ and the bias-correction constant $z_{\circ },$ then$%
\;\;$%
\begin{equation}
\hat{\Omega}_{j_{h}Boot\;-BC}^{\ast }=G^{-1}\left[ \Phi \left( 2z_{\circ
}+z_{\vartheta }\right) \right] .  \label{3.2.1}
\end{equation}%
Thus, the approximate bootstrap-BC $100(1-\vartheta )\%\;$confidence
interval of $\hat{\Omega}_{j_{h}Boot\;-BC}^{\ast }\;$is given by:
\begin{equation*}
\;%
\begin{bmatrix}
\hat{\Omega}_{j_{h}Boot\;-BC}^{\ast }(\frac{\vartheta }{2}) & , & \hat{\Omega%
}_{j_{h}Boot\;-BC}^{\ast }(1-\frac{\vartheta }{2})\;%
\end{bmatrix}%
.\;
\end{equation*}

\subsubsection{Bootstrap bias corrected accelerated confidence interval}

Let $\Phi (z)=\vartheta \;$be the standard normal cumulative distribution
function, with $z_{\zeta }=\Phi ^{-1}(\vartheta )\;$and the bias-correction
constant $z_{\circ }\;$which is defined in \eqref{3.2.0}. Then%
\begin{equation}
\hat{\Omega}_{j_{h}Boot\;-BCa}^{\ast }=G^{-1}\left[ \Phi \left( z_{\circ }+%
\frac{z_{\circ }+z_{\vartheta }}{1-a_{j_{h}}(z_{\circ }+z_{\vartheta })}%
\right) \right] ,j=1,2,3;\text{for }h=1,2,...,k.  \label{3.2.2}
\end{equation}
where $a_{j_{h}}\;$is called the acceleration factor which is estimated by a simple jack-knife method. Let $\underline{y}_{i}\;$represent the original data with the $i$th point omitted, say $\underline{y}_{2}=y_{1;r,N}<y_{3;r,N}<...<y_{r;r,N},$ and $\hat{\Omega}_{j_{h}}^{i}=\hat{\Omega}_{j}(\underline{y}_{i})$ be the estimate of $\Omega _{j_{h}}\;$ constructed from this data, $\Omega _{1_{h}}=\alpha _{h},\;\Omega_{2_{h}}=\theta _{h}$ and $\Omega _{3_{h}}=\beta _{h}$.$\;$Let $\bar{\Omega}_{j_{h}}\;$ be the mean of the $\hat{\Omega}_{j_{h}}^{i,}$s. Then $a_{j_{h}}$ is estimated by:
\begin{equation*}
a_{h}=\frac{\sum_{i=1}^{r}\left( \bar{\Omega}_{j_{h}}-\hat{\Omega}_{j_{h}}^{i,}\right) ^{3}}{6\left[ \sum_{i=1}^{r}\left( \bar{\Omega}_{j_{h}}-\hat{\Omega}_{j_{h}}^{i,}\right) ^{2}\right] ^{\frac{3}{2}}},j=1,2,3;\text{for }h=1,2,...,k.
\end{equation*}
For more details see \cite{26} and \cite{27}. If $a_{h}=0,$ equation %
\eqref{3.2.2} reduces to equation \eqref{3.2.1}. Then, the approximate
bootstrap-BC $100(1-\vartheta )\%\;$confidence interval of $\hat{\Omega}%
_{kBoot\;-BCa}^{\ast }\;$is given by: 
\begin{equation*}
\;%
\begin{bmatrix}
\hat{\Omega}_{j_{h}Boot\;-BCa}^{\ast }(\frac{\vartheta }{2}) & , & \hat{%
\Omega}_{j_{h}Boot\;-BCa}^{\ast }(1-\frac{\vartheta }{2})\;%
\end{bmatrix}%
.\;
\end{equation*}

\subsection{Comparison between the reliability of $k$ production lines}

To compare the reliablity of the units from $k$ production lines, the
following lemma is needed.

\begin{lemma}
If the random variable $X\sim WG(\alpha ,\theta ,\beta ),$ 
then $Y=\left( \frac{T}{\alpha }\right) ^{\theta }\sim WG(1,1,\beta ).$
\end{lemma}

\begin{proof}
The proof is easy to obtain. If a unit $U_{i}$ is selected randomly from the production line $i$ and another one $U_{j}$ \ is selected randomly from the production line $j,$ for $i,j=1,2,...,k$ and $i\neq j.$ Based on Equation \eqref{2.1.2}, the invariance property of the MLEs of the parameters and using Lemma 1, the reliability function for any two production lines $i$ and $j$ \ at time $t$ are 
\begin{equation}
R_{i}\left( t\right) =\left( 1+t\right) ^{-\hat{\beta}_{i}}\text{ and }%
R_{j}\left( t\right) =\left( 1+t\right) ^{-\hat{\beta}_{j}};i\;,j=1,2,...,k%
\text{ },\text{ }i\neq j.  \label{3.2.3}
\end{equation}
\end{proof}

According to the transformation stated in the previous lemma, we can prove the existence and uniqueness of the MLEs, for $\beta_h$ only, graphically, see Balakrishnan and Kateri \cite{28} for more details. Equation \eqref{3.5} can be rewritten as follows:
\begin{equation}
\frac{1}{\beta _{h}}=\frac{\sum_{i=1}^{r}\left( \delta_{i}(h)+s_{i}(h)\right) \log \left( 1+\left( \frac{\zeta _{i}}{\alpha _{h}}\right) ^{\theta _{h}}\right)}{M_{r}(h)}=\text{Constant}.  \label{3.2.4}
\end{equation}

\begin{corollary}
For any two production lines $i$ and $j$ and $\hat{%
\beta}_{i}\;\leq \;\hat{\beta}_{j},\;$then $R_{i}\left( t\right) \geq
R_{j}\left( t\right) $ at any time $t.$
\end{corollary}

\begin{figure}[h]
	\centering
    	\includegraphics[width=110mm,scale=1]{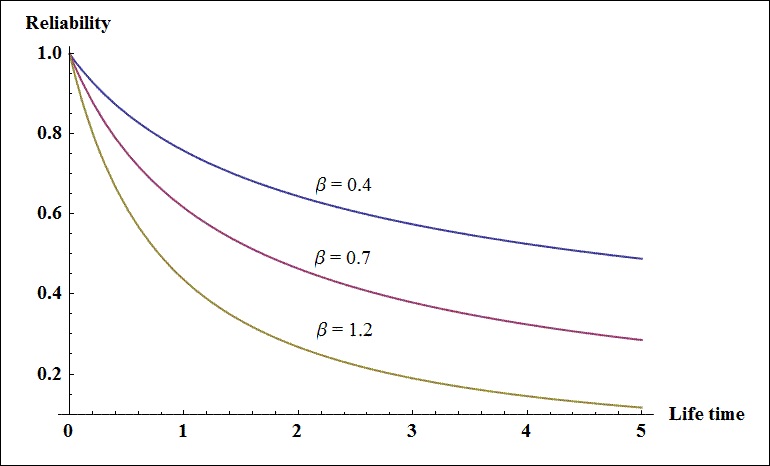}
	\caption{The effect of $\beta$ on the reliability function for some production lines.}
	\label{FIG:1}
\end{figure}

\section{Bayesian estimation based on different loss functions}\label{sec4}

Let the prior knowledge of the parameters $\alpha _{h},~\theta _{h}$ and $\beta_{h}$ be described by the following prior distributions :
\begin{equation}
\begin{array}{l}
\pi _{1}\left( \alpha _{h}\right) \propto \alpha _{h}^{\mu
_{h}-1}\;e^{-\alpha _{h}\lambda _{h}},\;\;\ \ \ \alpha _{h}>0, \\ 
\pi _{2}\left( \theta _{h}\right) \propto \theta _{h}^{q_{h}-1\;} ~e^{-\theta
_{h}w_{h}},\;\;\ \ \ \theta _{h}>0, \\ 
\pi _{3}\left( \beta _{h}\right) \propto \beta _{h}^{p_{h}-1}~~e^{-\beta
_{h}v_{h}},~~~~~\beta _{h}>0,%
\end{array} \text{~~~~~~~~where }\mu _{h},q_{h},p_{h,},\lambda _{h},w_{h},v_{h}>0,
\end{equation}
and $\alpha _{h},\theta _{h}$ and $\beta _{h}\;$are independent random variables for all $h=1,2,...,k.$ Many authors like Kundu and Howlader \cite{29}, Dey and Dey \cite{30} and Dey et. al. \cite{31} established the Bayesian estimation for their parameter models based on informative gamma priors. Mansour and Ramadan \cite{32} studied parameters estimation of the modified extended exponential distribution based on a Hybrid type II censoring scheme, they obtained the Bayesian estimates for non-informative priors and Gamma priors, and they concluded that Bayesian estimates for Gamma priors have the best performance.

Hence, the joint prior of the parameters $\alpha _{h},\theta _{h}$ and $%
\beta _{h}$ can be written as follows:%
\begin{equation}
\pi \left( \alpha _{h},\theta _{h},\beta _{h}\right) \propto \alpha
_{h}^{\mu _{h}-1}\;\theta _{h}^{q_{h}-1\;}\beta _{h}^{p_{h}-1}e^{-(\alpha
_{h}\lambda _{h}+\theta _{h}w_{h}+\beta _{h}v_{h})}.\;\;  \label{5.1}
\end{equation}%
The joint posterior density function of $\alpha _{h},\theta _{h}$ and $\beta
_{h},$ denoted by $\pi ^{\ast }(\alpha _{h},\theta _{h},\beta _{h}|\delta,s,\zeta)$ can be written as:
\begin{equation}
\pi ^{\ast }(\alpha _{h},\theta _{h},\beta _{h}|\delta,s,\zeta)=\frac{L\left( \alpha _{h},\theta _{h},\beta _{h}|\delta,s,\zeta\right) \times \pi \left( \alpha
_{h},\theta _{h},\beta _{h}\right) }{\int_{0}^{\infty }\int_{0}^{\infty
}\int_{0}^{\infty }L\left( \alpha _{h},\theta _{h},\beta _{h}|\delta,s,\zeta\right) \times \pi \left( \alpha _{h},\theta
_{h},\beta _{h}\right) d\alpha _{h}d\theta _{h}d\beta _{h}}.  \label{5.2}
\end{equation}
The joint posterior distribution, which combines the information in both the prior distribution and the likelihood, leads to containing more accurate information and getting a narrower range of possible values for the parameters. The Bayes estimate of any function of the parameters, say $g\left( \alpha _{h},\theta _{h},\beta _{h}\right), h=1,2,...,k,$ using SEL function is given by:
\begin{equation}
\hat{g}_{BS}\left( \alpha _{h},\theta _{h},\beta _{h}\right)=E_{\alpha_{h},\theta _{h},\beta _{h}|\delta,s,\zeta }\left[ g\left( \alpha _{h},\theta _{h},\beta _{h}\right) \right]=\frac{\int_{0}^{\infty }\int_{0}^{\infty }\int_{0}^{\infty }g\left(\alpha _{h},\theta _{h},\beta _{h}\right) \times L\left( \alpha _{h},\theta_{h},\beta _{h}\right) \times \pi \left( \alpha _{h},\theta _{h},\beta_{h}\right) d\alpha _{h}d\theta _{h}d\beta _{h}}{\int_{0}^{\infty}\int_{0}^{\infty }\int_{0}^{\infty }L\left( \alpha _{h},\theta _{h},\beta_{h}\right) \times \pi \left( \alpha _{h},\theta _{h},\beta _{h}\right)d\alpha _{h}d\theta _{h}d\beta _{h}}, \label{5.3}
\end{equation}
while the Bayes estimate of $g\left( \alpha _{h},\theta _{h},\beta_{h}\right) ,$ for all $h=1,2,...,k,$ using LINEX loss function is given by:
\begin{equation}
\hat{g}_{BL}\left( \alpha _{h},\theta _{h},\beta _{h}\right) =\frac{-1}{c}%
\log \left[ E_{\alpha _{h},\theta _{h},\beta _{h}|\delta,s,\zeta }\left[ e^{-c\;g\left( \alpha _{h},\theta _{h},\beta
_{h}\right) }\right] \right] ,\;c\neq 0,  \label{5.003}
\end{equation}

where
\begin{equation}
E_{\alpha _{h},\theta _{h},\beta _{h}|\delta,s,\zeta}\left[ e^{-cg\left( \alpha _{h},\theta _{h},\beta _{h}\right) }\right]=\frac{\int_{0}^{\infty }\int_{0}^{\infty }\int_{0}^{\infty}e^{-cg\left( \alpha _{h},\theta _{h},\beta _{h}\right) }\times L\left(\alpha _{h},\theta _{h},\beta _{h}\right) \times \pi \left( \alpha_{h},\theta _{h},\beta _{h}\right) d\alpha _{h}d\theta _{h}d\beta _{h}}{\int_{0}^{\infty }\int_{0}^{\infty }\int_{0}^{\infty }L\left( \alpha_{h},\theta _{h},\beta _{h}\right) \times \pi \left( \alpha _{h},\theta_{h},\beta _{h}\right) d\alpha _{h}d\theta _{h}d\beta _{h}}.
\label{5.0003}
\end{equation}

It is noticed that the ratio of the two integrals given by \eqref{5.3} and \eqref{5.0003} cannot be obtained in an explicit form. In this case, the MCMC technique is used to generate samples from the posterior distributions and then the Bayes estimates of the parameters $\alpha _{h},\theta _{h}$ and $\beta _{h},$ for all $h=1,2,...,k,$ will be computed. The main theme of the MCMC technique is to compute approximate values for the integrals in \eqref{5.3} and \eqref{5.0003}. An important sup-class of MCMC methods is Gibbs sampling and more general Metropolis within Gibbs samplers. The Metropolis algorithm is a random walk that uses an acceptance/rejection rule to converge to the target distribution. The Metropolis algorithm was first proposed by Metropolis et al. \cite{33} and was then generalized by Hastings \cite{34}. From \eqref{3.2}, \eqref{5.1} and \eqref{5.2}, the joint posterior density function of $\alpha _{h},\theta _{h}$ and $\beta _{h}$ , can be written as:
\begin{eqnarray}
\pi ^{\ast }(\alpha _{h},\theta _{h},\beta _{h}|\delta,s,\zeta)\propto \prod_{h=1}^{k}\frac{\theta_{h}^{M_{r}(h)+q_{h}-1}\beta _{h}^{M_{r}(h)+p_{h}-1}}{\alpha_{h}^{M_{r}(h)-\mu _{h}+1}}\times \exp \left\{ -(\alpha _{h}\lambda_{h}+\theta _{h}w_{h}+\beta _{h}v_{h})\right\} \notag \\ \times \exp \left\{ -\sum_{i=1}^{r}\left[ \left( \beta _{h}\left( \delta_{i}(h)+s_{i}(h)\right) +\delta _{i}(h)\right) \log \left( 1+\left( \frac{\zeta _{i}}{\alpha _{h}}\right) ^{\theta _{h}}\right) \right] \right\}\times \notag \\ \exp \left\{ \left( \theta _{h}-1\right) \left[ \sum_{i=1}^{r}\delta_{i}(h)\log \zeta _{i}-M_{r}(h)\log \alpha _{h}\right] \right\} .
\label{5.04}
\end{eqnarray}

The conditional posterior densities of $\alpha _{h},\theta _{h}$ and $\beta_{h}$ can also be written as:%
\begin{eqnarray}
\pi _{1}^{\ast }(\alpha _{h}|\theta _{h},\beta _{h},\delta,s,\zeta)\propto \prod_{h=1}^{k}\alpha _{h}^{\mu _{h}-M_{r}(h)-1} \times \exp \left\{ -\left( \theta _{h}-1\right) \left[ M_{r}(h)\log \alpha_{h}\right] -\alpha _{h}\lambda _{h}\right\} \notag \\ \times \exp \left\{ -\sum_{i=1}^{r}\left[ \left( \beta _{h}\left( \delta_{i}(h)+s_{i}(h)\right) +\delta _{i}(h)\right) \log \left( 1+\left( \frac{\zeta _{i}}{\alpha _{h}}\right) ^{\theta _{h}}\right) \right] \right\}
\label{5.5}
\end{eqnarray}

\begin{eqnarray}
\pi _{2}^{\ast }(\theta _{h}|\alpha _{h},\beta _{h},\delta,s,\zeta)\propto \prod_{h=1}^{k}\theta _{h}^{M_{r}(h)+q_{h}-1} \times \exp \left\{ -\theta _{h}\left[ w_{h}+M_{r}(h)\log \alpha_{h}-\sum_{i=1}^{r}\delta _{i}(h)\log \zeta _{i}\right] \right\} \notag \\ \times \exp \left\{ -\sum_{i=1}^{r}\left[ \left( \beta _{h}\left( \delta_{i}(h)+s_{i}(h)\right) +\delta _{i}(h)\right) \log \left( 1+\left( \frac{\zeta _{i}}{\alpha _{h}}\right) ^{\theta _{h}}\right) \right] \right\}
\label{5.6}
\end{eqnarray}

\begin{equation}
\pi _{3}^{\ast }(\beta _{h}|\alpha _{h},\theta _{h},\delta,s,\zeta)\equiv gamma\left[ M_{r}(h)+p_{h},v_{h}+\sum_{i=1}^{r}\left[ \left( \delta _{i}(h)+s_{i}(h)\right) \log \left( 1+\left( \frac{\zeta _{i}}{\alpha _{h}}\right) ^{\theta _{h}}\right) \right] \right]
\label{5.7}
\end{equation}

Now, the following steps illustrate the method of the Metropolis--Hastings algorithm within Gibbs sampling to generate the posterior samples as suggested by Tierney \cite{35}, and then the Bayes estimates and the corresponding credible intervals can be obtained:

\begin{enumerate}
\item[(1)] Start with $\left( \alpha _{h}^{\left( 0\right) }=\hat{\alpha}%
_{h},\;\theta _{h}^{\left( 0\right) }=\hat{\theta}_{h}\text{ and }\beta
_{h}^{\left( 0\right) }=\hat{\beta}_{h}\right) .$

\item[(2)] Put $l=1.$

\item[(3)] Generate $\beta _{h}^{(l)}$ from 
\begin{equation*}
gamma\left[ M_{r}(h)+p_{h},v_{h}+\sum_{i=1}^{r}\left[ \left( \delta
_{i}(h)+s_{i}(h)\right) \log \left( 1+\left( \frac{\zeta _{i}}{\alpha _{h}}%
\right) ^{\theta _{h}}\right) \right] \right]
\end{equation*}

\item[(4)] Using the following Metropolis-Hastings method, generate $\alpha
_{h}^{\left( l\right) }$ and $\theta _{h}^{\left( l\right) }$ from %
\eqref{5.5} and \eqref{5.6} with the suggested normal distributions%
\begin{equation*}
N(\alpha _{h}^{(l-1)},var\left( \alpha _{h}\right) )\text{ and }N(\theta
_{h}^{(l-1)},var\left( \theta _{h}\right) ),\text{respectively,}
\end{equation*}%
where $var\left( \alpha _{h}\right) $ and $var\left( \theta _{h}\right) $
can be obtained from the main diagonal in asymptotic inverse Fisher
information matrix \eqref{3.7}.

\begin{description}

\item[i-] Generate a proposal $\alpha _{h}^{\ast }$ from $N(\alpha
_{h}^{(l-1)},var\left( \alpha _{h}\right) )$ and $\theta _{h}^{\ast }$ from $%
N(\theta _{h}^{(l-1)},var\left( \theta _{h}\right) ).$

\item[ii-] Evaluate the acceptance probabilities

\begin{equation}
\begin{array}{l}
\rho _{\alpha _{h}}=\min \left[ 1,\frac{\pi _{1}^{\ast }(\alpha _{h}^{\ast
}|\theta _{h}^{\left( l-1\right) },\beta _{h}^{\left( l\right) },\delta,s,\zeta)}{\pi _{1}^{\ast }(\alpha _{h}^{\left(
l-1\right) }|\theta _{h}^{\left( l-1\right) },\beta ^{\left( l\right) },\delta,s,\zeta)}\right] , \\ 
\\ 
\rho _{\theta _{h}}=\min \left[ 1,\frac{\pi _{2}^{\ast }(\theta _{h}^{\ast
}|\alpha _{h}^{\left( l\right) },\beta _{h}^{\left( l\right) },\delta,s,\zeta}{\pi _{2}^{\ast }(\theta _{h}^{\left(
l-1\right) }|\alpha _{h}^{\left( l\right) },\beta _{h}^{\left( l\right) },\delta,s,\zeta}\right] \\ 
\end{array}.
\end{equation}

\item[iii-] Generate $u_{1}$ and $u_{2}$ from a Uniform $(0,1)$ distribution.

\item[iv-] If $u_{1}\leq $ $\rho _{\alpha _{h}}$ , then accept the proposal
and set $\alpha _{h}^{\left( l\right) }=\alpha _{h}^{\ast },$ else set $%
\alpha _{h}^{\left( l\right) }=\alpha _{h}^{\left( l-1\right) }.$

\item[v-] If $u_{2}\leq $ $\rho _{\theta _{h}}$ , then accept the proposal
and set $\theta _{h}^{\left( l\right) }=\theta _{h}^{\ast },$ else set $%
\theta _{h}^{\left( l\right) }=\theta _{h}^{\left( l-1\right) }.$

\end{description}

\item[(5)] Compute $\alpha _{h}^{\left( l\right) }$ and $\theta _{h}^{\left(
l\right) }.$

\item[(6)] Put $l=l+1.$

\item[(7)] Repeat Steps $3-6$ $Q$ times.

\item[(8)] In order to guarantee the convergence and to remove the influence of the initial values selection, the first $M$ simulated points are ignored. The selected samples are $\alpha _{h}^{\left( l\right) }$ and $\theta _{h}^{\left( l\right) },$ $l=M+1,...,Q,$ for sufficiently large $Q$. The approximate Bayes estimates for $\alpha _{h},\theta _{h}$ and $\beta_{h} $ based on SEL are:

\begin{equation}
\begin{array}{c}
\alpha _{hBS}=\frac{1}{Q-M}\sum_{l=M+1}^{Q}\alpha _{h}^{\left( l\right) },
\\ 
\\ 
\theta _{hBS}=\frac{1}{Q-M}\sum_{i=M+1}^{Q}\theta _{h}^{\left( l\right) },
\\ 
\\
\beta _{hBS}=\frac{1}{Q-M}\sum_{i=M+1}^{Q}\beta _{h}^{\left( l\right) }.%
\end{array},
\end{equation}

and the estimates for the same parameters under LINEX loss function are

\begin{equation}
\begin{array}{c}
\alpha _{hBL}=\frac{-1}{c}\log \left[ \frac{1}{Q-M}\sum_{l=M+1}^{Q}e^{-c\alpha _{h}^{\left( l\right) }}\right] , \\ 
\\ 
\theta _{hBL}=\frac{-1}{c}\log \left[ \frac{1}{Q-M}\sum_{l=M+1}^{Q}e^{-c\theta _{h}^{\left( l\right) }}\right] , \\ \\ 
\beta _{hBL}=\frac{-1}{c}\log \left[ \frac{1}{Q-M}\sum_{l=M+1}^{Q}e^{-c\beta _{h}^{\left( l\right) }}\right] .
\end{array}.
\end{equation}

\item[(9)] To calculate the credible intervals (CRIs) of $\Omega _{j_{h}}$, where $\Omega _{1_{h}}=\alpha _{h},\;\Omega _{2_{h}}=\theta _{h}$ and $\Omega
_{3_{h}}=\beta _{h},$ the quantiles of the sample are assumed to be the endpoints of the intervals. Sort $\left\{ \Omega _{j_{h}}^{M+1},\Omega_{j_{h}}^{M+2},...,\Omega _{j_{h}}^{Q}\right\} $ as $\left\{ \Omega_{j_{h}}^{\left( 1\right) },\Omega _{j_{h}}^{\left( 2\right) },...,\Omega_{j_{h}}^{\left( Q-M\right) }\right\} ,j=1,2,3;$for $h=1,2,...,k.$ Hence the $100\left( 1-\vartheta \right)\%$ symmetric credible interval of $\Omega
_{j_{h}}$ is given by:
\end{enumerate}
\begin{equation}
\begin{bmatrix}
\Omega _{j_{h}}\left( \frac{\vartheta }{2}\left( Q-M\right) \right) & , & 
\Omega _{j_{h}}\left( \left( 1-\frac{\vartheta }{2}\right) \left( Q-M\right)
\right)%
\end{bmatrix}%
.
\end{equation}

\section{Simulation Study}\label{sec5}
The prior knowledge of the number of failures before the test procedure plays a meaningful role in determining the appropriate sampling plans and developing good estimators for parameters. In this section, a simulation study is performed under different JPROG-II-C schemes where 1000 JPROG-II-C samples are generated from the two WG populations, based on distinct sample sizes, for calculating the approximation of the expected values of the number of failures from the first production line before the test procedure (A.E.B). The mean of the exact number of failures after the test procedure (M.E.A) is also computed from the first production line. The comparison between the A.E.B and M.E.A is necessary to judge the effectiveness of the approximation introduced by Parsi and Bairamov \cite{36} throughout our model. Parsi and Bairamov \cite{36} used their approximation in the case of JPROG-II-C samples generated from two Weibull Gamma and two Pareto populations. They found that the results of A.E.B and M.E.A were close to each other in most cases. After applying this approximation to the simulated JPROG-II-C samples generated from two WG populations, we reached the same results introduced by Parsi and Bairamov and found that they are close together in most cases. In addition to these close results, the A.E.B and M.E.A. values seem to be the same for large values of $N$. The calculations for our model are performed based on the following assumptions, $X_{1}$ is chosen from $WG(6,3,4)$ and $X_{2}$ from $WG(2,5,1.5),~p=0.13$. Once again $X_{1}$ is chosen to from $WG(2,5,1.5)$ and $X_{2}$ from $WG(6,3,4),~p=0.87,$ where $p=P(X_{1}<X_{2})=\int_{0}^{\infty }\bar{G}(x;\alpha _{2},\theta _{2},\beta_{2})dF(x;\alpha _{1},\theta _{1},\beta _{1}),~X_{1}$ and $X_{2}$ represent the lifetime of the first production line units and the second production line units respectively. Also $\bar{G}(x;\alpha _{2},\theta _{2},\beta _{2})$ is the survival function of $X_{2}$, while $F(x;\alpha _{1},\theta_{1},\beta _{1})$ is the CDF of $X_{1}.$

\begin{table}[!]
    \centering
\begin{tabular}{ccccccccc}
\multicolumn{9}{c}{Table 1. The Comparison between A.E.B and M.E.A for the first production line.} \\ \hline
&  &  &  &  & \multicolumn{2}{c}{$p=0.13$} & \multicolumn{2}{c}{$p=0.87$} \\ 
\cline{6-9}
Scheme No. & $N$ & $r$ & Scheme ($R$) & $(n_{1},n_{2})$ & A.E.B & M.E.A & 
A.E.B & M.E.A \\ \hline
\multicolumn{1}{l}{$%
\begin{array}{c}
1 \\ 
2 \\ 
3%
\end{array}%
$} & \multicolumn{1}{l}{$30$} & \multicolumn{1}{l}{$10$} & $(0,...,0,20)$ & 
\multicolumn{1}{l}{$%
\begin{array}{c}
(20,10) \\ 
(15,15) \\ 
(10,20)%
\end{array}%
$} & $%
\begin{array}{c}
3.091 \\ 
1.667 \\ 
0.858%
\end{array}%
$ & $%
\begin{array}{c}
3.028 \\ 
1.664 \\ 
0.930%
\end{array}%
$ & $%
\begin{array}{c}
9.142 \\ 
8.332 \\ 
6.909%
\end{array}%
$ & $%
\begin{array}{c}
9.109 \\ 
8.319 \\ 
6.974%
\end{array}%
$ \\ 
\multicolumn{1}{l}{$%
\begin{array}{c}
4 \\ 
5 \\ 
6%
\end{array}%
$} & \multicolumn{1}{l}{} & \multicolumn{1}{l}{} & $(0,0,0,0,10,10,0,0,0,0)$
& \multicolumn{1}{l}{$%
\begin{array}{c}
(20,10) \\ 
(15,15) \\ 
(10,20)%
\end{array}%
$} & $%
\begin{array}{c}
5.629 \\ 
9.799 \\ 
2.908%
\end{array}%
$ & $%
\begin{array}{c}
4.807 \\ 
3.477 \\ 
2.198%
\end{array}%
$ & $%
\begin{array}{c}
9.334 \\ 
8.255 \\ 
14.192%
\end{array}%
$ & $%
\begin{array}{c}
7.832 \\ 
6.656 \\ 
5.250%
\end{array}%
$ \\ 
\multicolumn{1}{l}{$%
\begin{array}{c}
7 \\ 
8 \\ 
9%
\end{array}%
$} & \multicolumn{1}{l}{} & \multicolumn{1}{l}{} & $(20,0,...,0)$ & 
\multicolumn{1}{l}{$%
\begin{array}{c}
(20,10) \\ 
(15,15) \\ 
(10,20)%
\end{array}%
$} & $%
\begin{array}{c}
6.348 \\ 
4.742 \\ 
3.211%
\end{array}%
$ & $%
\begin{array}{c}
6.506 \\ 
4.800 \\ 
3.243%
\end{array}%
$ & $%
\begin{array}{c}
6.789 \\ 
5.258 \\ 
3.652%
\end{array}%
$ & $%
\begin{array}{c}
6.776 \\ 
5.137 \\ 
3.479%
\end{array}%
$ \\ 
\multicolumn{1}{l}{$%
\begin{array}{c}
10 \\ 
11 \\ 
12%
\end{array}%
$} & \multicolumn{1}{l}{} & \multicolumn{1}{l}{$15$} & $(0,...,0,15)$ & 
\multicolumn{1}{l}{$%
\begin{array}{c}
(20,10) \\ 
(15,15) \\ 
(10,20)%
\end{array}%
$} & $%
\begin{array}{c}
5.829 \\ 
3.056 \\ 
1.508%
\end{array}%
$ & $%
\begin{array}{c}
5.825 \\ 
2.926 \\ 
1.391%
\end{array}%
$ & $%
\begin{array}{c}
13.492 \\ 
11.944 \\ 
9.171%
\end{array}%
$ & $%
\begin{array}{c}
13.582 \\ 
12.140 \\ 
9.137%
\end{array}%
$ \\ 
\multicolumn{1}{l}{$%
\begin{array}{c}
13 \\ 
14 \\ 
15%
\end{array}%
$} & \multicolumn{1}{l}{} & \multicolumn{1}{l}{} & $(15,0,...,0)$ & 
\multicolumn{1}{l}{$%
\begin{array}{c}
(20,10) \\ 
(15,15) \\ 
(10,20)%
\end{array}%
$} & $%
\begin{array}{c}
9.785 \\ 
7.177 \\ 
4.444%
\end{array}%
$ & $%
\begin{array}{c}
9.846 \\ 
7.337 \\ 
4.841%
\end{array}%
$ & $%
\begin{array}{c}
10.556 \\ 
7.823 \\ 
5.215%
\end{array}%
$ & $%
\begin{array}{c}
10.056 \\ 
7.605 \\ 
5.165%
\end{array}%
$ \\ 
\multicolumn{1}{l}{$%
\begin{array}{c}
16 \\ 
17 \\ 
18%
\end{array}%
$} & \multicolumn{1}{l}{} & \multicolumn{1}{l}{$20$} & $(0,...,0,10)$ & 
\multicolumn{1}{l}{$%
\begin{array}{c}
(20,10) \\ 
(15,15) \\ 
(10,20)%
\end{array}%
$} & $%
\begin{array}{c}
10.039 \\ 
5.508 \\ 
2.545%
\end{array}%
$ & $%
\begin{array}{c}
10.153 \\ 
5.683 \\ 
2.421%
\end{array}%
$ & $%
\begin{array}{c}
17.455 \\ 
14.492 \\ 
9.961%
\end{array}%
$ & $%
\begin{array}{c}
17.532 \\ 
14.347 \\ 
9.809%
\end{array}%
$ \\ 
\multicolumn{1}{l}{$%
\begin{array}{c}
19 \\ 
20 \\ 
21%
\end{array}%
$} & \multicolumn{1}{l}{} & \multicolumn{1}{l}{} & $(10,0,...,0)$ & 
\multicolumn{1}{l}{$%
\begin{array}{c}
(20,10) \\ 
(15,15) \\ 
(10,20)%
\end{array}%
$} & $%
\begin{array}{c}
13.184 \\ 
9.866 \\ 
9.863%
\end{array}%
$ & $%
\begin{array}{c}
13.216 \\ 
9.886 \\ 
6.615%
\end{array}%
$ & $%
\begin{array}{c}
10.137 \\ 
10.134 \\ 
6.816%
\end{array}%
$ & $%
\begin{array}{c}
13.331 \\ 
10.113 \\ 
6.796%
\end{array}%
$ \\ 
\multicolumn{1}{l}{$%
\begin{array}{c}
22 \\ 
23 \\ 
24%
\end{array}%
$} & \multicolumn{1}{l}{$50$} & \multicolumn{1}{l}{$20$} & $(0,...,0,30)$ & 
\multicolumn{1}{l}{$%
\begin{array}{c}
(30,20) \\ 
(25,25) \\ 
(20,30)%
\end{array}%
$} & $%
\begin{array}{c}
5.296 \\ 
3.624 \\ 
2.440%
\end{array}%
$ & $%
\begin{array}{c}
4.980 \\ 
3.391 \\ 
2.389%
\end{array}%
$ & $%
\begin{array}{c}
17.560 \\ 
16.376 \\ 
14.704%
\end{array}%
$ & $%
\begin{array}{c}
17.669 \\ 
16.546 \\ 
14.974%
\end{array}%
$ \\ 
\multicolumn{1}{l}{$%
\begin{array}{c}
25 \\ 
26 \\ 
27%
\end{array}%
$} & \multicolumn{1}{l}{} & \multicolumn{1}{l}{} & $(30,...,0,0)$ & 
\multicolumn{1}{l}{$%
\begin{array}{c}
(30,20) \\ 
(25,25) \\ 
(20,30)%
\end{array}%
$} & $%
\begin{array}{c}
11.749 \\ 
9.761 \\ 
7.679%
\end{array}%
$ & $%
\begin{array}{c}
11.897 \\ 
9.826 \\ 
7.875%
\end{array}%
$ & $%
\begin{array}{c}
12.321 \\ 
10.239 \\ 
8.251%
\end{array}%
$ & $%
\begin{array}{c}
12.073 \\ 
10.033 \\ 
8.202%
\end{array}%
$ \\ 
\multicolumn{1}{l}{$%
\begin{array}{c}
28 \\ 
29 \\ 
30%
\end{array}%
$} & \multicolumn{1}{l}{} & \multicolumn{1}{l}{$25$} & $(0,...,0,25)$ & 
\multicolumn{1}{l}{$%
\begin{array}{c}
(30,20) \\ 
(25,25) \\ 
(20,30)%
\end{array}%
$} & $%
\begin{array}{c}
7.631 \\ 
5.154 \\ 
3.413%
\end{array}%
$ & $%
\begin{array}{c}
7.457 \\ 
4.844 \\ 
3.156%
\end{array}%
$ & $%
\begin{array}{c}
21.587 \\ 
19.846 \\ 
17.369%
\end{array}%
$ & $%
\begin{array}{c}
21.690 \\ 
20.210 \\ 
17.576%
\end{array}%
$ \\ 
\multicolumn{1}{l}{$%
\begin{array}{c}
31 \\ 
32 \\ 
33%
\end{array}%
$} & \multicolumn{1}{l}{} & \multicolumn{1}{l}{} & $(25,0,...,0)$ & 
\multicolumn{1}{l}{$%
\begin{array}{c}
(30,20) \\ 
(25,25) \\ 
(20,30)%
\end{array}%
$} & $%
\begin{array}{c}
14.788 \\ 
12.316 \\ 
9.789%
\end{array}%
$ & $%
\begin{array}{c}
14.842 \\ 
12.428 \\ 
9.941%
\end{array}%
$ & $%
\begin{array}{c}
15.211 \\ 
12.684 \\ 
10.212%
\end{array}%
$ & $%
\begin{array}{c}
15.138 \\ 
12.560 \\ 
10.070%
\end{array}%
$ \\ 
\multicolumn{1}{l}{$%
\begin{array}{c}
34 \\ 
35 \\ 
36%
\end{array}%
$} & \multicolumn{1}{l}{$70$} & \multicolumn{1}{l}{$30$} & $(0,...,0,40)$ & 
\multicolumn{1}{l}{$%
\begin{array}{c}
(40,30) \\ 
(35,35) \\ 
(30,40)%
\end{array}%
$} & $%
\begin{array}{c}
7.437 \\ 
5.651 \\ 
4.258%
\end{array}%
$ & $%
\begin{array}{c}
7.058 \\ 
5.306 \\ 
4.129%
\end{array}%
$ & $%
\begin{array}{c}
25.742 \\ 
24.349 \\ 
22.563%
\end{array}%
$ & $%
\begin{array}{c}
25.906 \\ 
24.734 \\ 
22.976%
\end{array}%
$ \\ 
\multicolumn{1}{l}{$%
\begin{array}{c}
37 \\ 
38 \\ 
39%
\end{array}%
$} & \multicolumn{1}{l}{} & \multicolumn{1}{l}{} & $(40,...,0,0)$ & 
\multicolumn{1}{l}{$%
\begin{array}{c}
(40,30) \\ 
(35,35) \\ 
(30,40)%
\end{array}%
$} & $%
\begin{array}{c}
16.909 \\ 
14.787 \\ 
12.673%
\end{array}%
$ & $%
\begin{array}{c}
17.100 \\ 
14.763 \\ 
12.761%
\end{array}%
$ & $%
\begin{array}{c}
17.327 \\ 
15.213 \\ 
13.091%
\end{array}%
$ & $%
\begin{array}{c}
17.209 \\ 
15.083 \\ 
12.824%
\end{array}%
$ \\ 
\multicolumn{1}{l}{$%
\begin{array}{c}
40 \\ 
41 \\ 
42%
\end{array}%
$} & \multicolumn{1}{l}{} & \multicolumn{1}{l}{$35$} & $(0,...,0,35)$ & 
\multicolumn{1}{l}{$%
\begin{array}{c}
(40,30) \\ 
(35,35) \\ 
(30,40)%
\end{array}%
$} & $%
\begin{array}{c}
9.612 \\ 
7.251 \\ 
5.417%
\end{array}%
$ & $%
\begin{array}{c}
9.224 \\ 
6.916 \\ 
5.033%
\end{array}%
$ & $%
\begin{array}{c}
29.583 \\ 
27.749 \\ 
25.388%
\end{array}%
$ & $%
\begin{array}{c}
29.918 \\ 
28.146 \\ 
25.779%
\end{array}%
$ \\ 
\multicolumn{1}{l}{$%
\begin{array}{c}
43 \\ 
44 \\ 
45%
\end{array}%
$} & \multicolumn{1}{l}{} & \multicolumn{1}{l}{} & $(35,...,0,0)$ & 
\multicolumn{1}{l}{$%
\begin{array}{c}
(40,30) \\ 
(35,35) \\ 
(30,40)%
\end{array}%
$} & $%
\begin{array}{c}
19.795 \\ 
17.313 \\ 
14.835%
\end{array}%
$ & $%
\begin{array}{c}
19.896 \\ 
17.391 \\ 
14.858%
\end{array}%
$ & $%
\begin{array}{c}
20.165 \\ 
17.687 \\ 
15.205%
\end{array}%
$ & $%
\begin{array}{c}
20.123 \\ 
17.576 \\ 
15.050%
\end{array}%
$ \\ \hline
\end{tabular}
\end{table}

It is noted that the A.E.B and M.E.A. values are close together in all schemes except for schemes No. 4, 5, and 6. Also, for large values of $N$, the A.E.B and M.E.A. values seem to be the same.

\section{Applications}\label{sec6}
In this section, our model is applied in the case of two WG populations. First we will show the results through a simulated example, then we will apply on a real data example.

\subsection{Simulated example}

An approximation is used for calculating the expected values of the number of failures before the test procedure. The exact values of the number of failures are also observed after the test procedure. The approximated values are compared with the exact values through a simulated example. A simulated data is generated in the case of $k=2$ and the following algorithm shows how the simulated JPROG-II-C data is generated from the two production lines.

\begin{enumerate}
\item[(1)] Generate a random sample with size $n_{1}$ from $WG(\alpha
_{1},\theta _{1},\beta _{1}).$

\item[(2)] Generate another random sample with size $n_{2}$ from $WG(\alpha
_{2},\theta _{2},\beta _{2}),$ such that the two samples, with sizes $n_{1}$
and $n_{2}$ are independent.

\item[(3)] Combine the two samples in one sample with size $N=n_{1}+n_{2}.$

\item[(4)] Sort the joint sample.

\item[(5)] Specify the censoring size $r$ and the censoring scheme $R_{i}$ which represent the number of the removed units at the time of the $i$th failure from the joint sample, $i=1,2,...,r.$

\item[(6)] Specify $s_{i}(1)$ and $s_{i}(2)$ which represent the number of the removed units at the time of the $i$th failure from the samples No.$1$ and No.2 respectively. So, $R_{i}=\sum_{h=1}^{2}s_{i}(h),$ $i=1,2,...,r.$

\item[(7)] Using the famed algorithm described in Balakrishnan and Sandhu 
\cite{37} for generating progressive Type -II samples, the ordered
JPROG-II-C is:
\begin{equation*}
\zeta _{1;r,N}^{\mathbf{R}}<\zeta _{2;r,N}^{\mathbf{R}}<...<\zeta _{r;r,N}^{\mathbf{R}}.
\end{equation*}

\item[(8)] Define the two indicators $\delta _{i}(1)$ and $\delta _{i}(2)$
that determine if the $i$th failure belongs to sample No.1 or
sample No.2 respectively such that $\sum_{h=1}^{2}\delta _{i}(h)=1,$ for $%
i=1,2,...,r.$

\item[(9)] Use the JPROG-II-C sample to compute the MLEs for the model
parameters. The Newton--Raphson method is applied for solving the nonlinear
system to obtain the MLEs of the parameters.

\item[(10)] Compute the 95\% bootstrap confidence intervals for the model
parameters, using the steps described in Section 3.

\item[(11)] Compute the Bayes estimates of the model parameters based on MCMC algorithm described in Section 4.
\end{enumerate}

\begin{center}
\begin{tabular}{llllllll}
\multicolumn{8}{c}{Table 2. Simulated data from the first production line}
\\ \hline
$3.8437$ & $4.7634$ & $4.8869$ & $5.3066$ & $5.7811$ & $5.8661$ & $6.2278$ & 
$6.7398$ \\ \hline
$4.318$ & $4.7833$ & $5.0482$ & $5.3785$ & $5.8257$ & $5.9144$ & $6.5284$ & $%
6.7949$ \\ \hline
$4.677$ & $4.8032$ & $5.1756$ & $5.6429$ & $5.853$ & $5.9308$ & $6.6121$ & $%
7.3244$ \\ \hline
\end{tabular}
\end{center}

Table 2 shows the simulated failure times from the first production line where the initial values for the parameters are chosen as follows: $\alpha _{1}=7.2,~\theta _{1}=9,~\beta _{1}=5.5$ with sample size $n_{1}=24.$

\begin{center}
\begin{tabular}{llllllll}
\multicolumn{8}{c}{Table 3. Simulated data from the second production line}
\\ \hline
$1.6875$ & $7.3303$ & $11.09$ & $37.016$ & $112.776$ & $496.943$ & $8761.22$
& $33040.6$ \\ \hline
$1.8412$ & $7.3971$ & $11.7046$ & $59.0379$ & $139.534$ & $764.616$ & $%
21450.3$ & $39820.6$ \\ \hline
$4.497$ & $10.4974$ & $12.5864$ & $74.212$ & $172.122$ & $6157.02$ & $%
30766.1 $ &  \\ \hline
\end{tabular}
\end{center}

Table 3 shows the simulated failure times from the production line No. 2 where the initial values for the parameters are chosen in this case as follows: $\alpha _{2}=2,~\theta _{2}=3,~\beta _{2}=0.06$ with sample size $n_{2}=23.$ It is clear that the value of $\beta _{2}$ is less than the value of $\beta _{1}$. So from Corollary 1, the units manufactured by the second production line have higher reliability than those manufactured by the first production line. Also, the difference in Table 2 and Table 3 is easily noticed, in addition to the big difference between failure times in Table 3 and the small difference between failure times in Table 2 due to the effect of $\beta _{1}$ and $\beta _{2}$ values on the units' reliability in each production line, and the random withdrawal of the removed units which is a specific identity for progressive censoring.

\begin{center}
\begin{tabular}{llllllllll}
\multicolumn{10}{c}{Table 4. Simulated JPROG-II-C data from the two
production lines} \\ \hline
$1.6875$ & $1.8412$ & $3.8437$ & $4.318$ & $4.497$ & $4.677$ & $4.7634$ & $%
4.8032$ & $4.8869$ & $5.0482$ \\ \hline
\end{tabular}
\end{center}

Table 4 shows the simulated JPROG-II-C data from the two production lines with joint progressive censoring schemes $N=n_{1}+n_{2},~r=10$. $R=\left(2,2,0,2,2,1,1,2,0,25\right) $ is chosen based on the choice of $s(1)$ and $s(2),$ where $s(1)=\left( 2,0,0,0,2,0,1,2,0,9\right)$ and $s(2)=\left(0,2,0,2,0,1,0,0,0,16\right)$, such that $R_{i}=\sum_{h=1}^{2}s_{i}(h),~i=1,2,...,10. $

\begin{center}
\begin{tabular}{cccccccc}
\multicolumn{8}{c}{Table 5. Different point estimates for $\alpha
_{1},\theta _{1},\beta _{1};\alpha _{2},\theta _{2},\beta _{2}.$} \\ \hline
& $\left( .\right) _{ML}$ & $\left( .\right) _{Boot\;-p}$ & $\left( .\right)
_{Boot\;-t}$ & $\left( .\right) _{_{BS}}$ & \multicolumn{3}{c}{$\left(
.\right) _{BL}$} \\ \cline{6-8}
& \ \ \ \ \ \ \  & \ \ \ \ \ \ \ \ \ \ \ \  & \ \ \ \ \ \ \ \ \ \  & \  & $%
c=0.0001$ & $c=-2$ & $c=2$ \\ \hline
$\hat{\alpha}_{1}$ & $6.3647$ & $6.320$ & $6.1146$ & $6.4491$ & $6.4491$ & $%
6.4495$ & $6.4488$ \\ \hline
$\hat{\theta}_{1}$ & $10.6394$ & $11.2044$ & $9.8411$ & $10.575$ & $10.575$
& $10.5773$ & $10.572$ \\ \hline
$\hat{\beta}_{1}$ & $5.735$ & $5.5335$ & $4.2986$ & $2.2711$ & $2.2711$ & $%
3.5636$ & $1.7479$ \\ \hline
$\hat{\alpha}_{2}$ & $1.4707$ & $1.7095$ & $1.0192$ & $1.4551$ & $1.4551$ & $%
1.4552$ & $1.455$ \\ \hline
$\hat{\theta}_{2}$ & $3.9635$ & $3.381$ & $3.6432$ & $3.9849$ & $3.9849$ & $%
3.9985$ & $3.9722$ \\ \hline
$\hat{\beta}_{2}$ & $0.0297$ & $0.2316$ & $0.0122$ & $0.0290$ & $0.0290$ & $%
0.0292$ & $0.0287$ \\ \hline
\end{tabular}
\end{center}

Also, graphical representations that show the existence and uniqueness of $\beta_1$ and $\beta_2$ based on Equation \eqref{3.2.4} are shown in Figure 2 (a) and (b) respectively.

\begin{figure}
     \centering
     \begin{subfigure}[b]{0.3\textwidth}
         \centering
         \includegraphics[width=\textwidth]{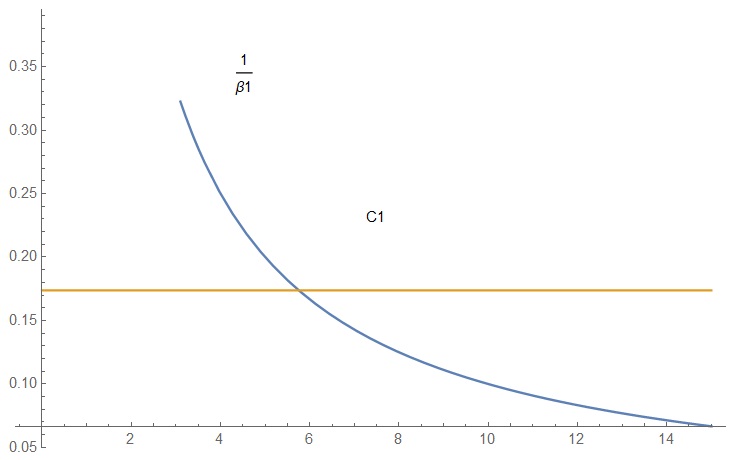}
         \caption{(a) First production line}
         \label{fig:Figure 2 (a)}
     \end{subfigure}
     \hspace{1cm}
     \begin{subfigure}[b]{0.3\textwidth}
         \centering
         \includegraphics[width=\textwidth]{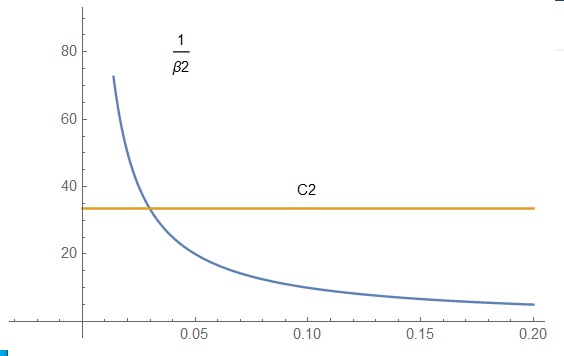}
         \caption{Second production line}
         \label{fig:Figure 2 (b)}
     \end{subfigure}
        \caption{Graphical illustration for the existence and uniqueness of the MLEs in the simulated data example.}
        \label{fig:Figure 2}
\end{figure}

In the MCMC approach, we run the chain $52000$ times and discard the first $2000$ values as `burn-in' while the prior knowledge parameters are chosen to be $\mu _{1}=$ $\mu _{2}=$ $q_{1}=q_{2}=p_{1}=p_{2}=0.01$ and $\lambda _{1}=$ $\lambda _{2}=w_{1}=w_{2}=v_{1}=v_{2}=2.$

\begin{center}
\begin{tabular}{ccccc}
\multicolumn{5}{c}{Table 6. 95\% confidence intervals for $\alpha
_{1},\theta _{1},\beta _{1};\alpha _{2},\theta _{2},\beta _{2}.\smallskip $}
\\ \hline
Method & $\alpha _{1}$ & Length & $\alpha _{2}$ & Length \\ \hline
ACI & $\left[ 1.689,11.041\right] $ & $9.35161$ & $\left[ 0.048,2.8935\right]
$ & $2.84556$ \\ \hline
Boot\ -p CI & $\left[ 5.1195,7.343\right] $ & $2.22353$ & $\left[
0.5345,5.097\right] $ & $4.56250$ \\ \hline
Boot\ -t CI & $\left[ 5.6166,6.276\right] $ & $0.65942$ & $\left[
-2.221,1.416\right] $ & $3.63759$ \\ \hline
Boot-BC CI & $\left[ 5.2265,7.386\right] $ & $2.15963$ & $\left[ 0.868,5.4262%
\right] $ & $4.55778$ \\ \hline
Boot-BCa CI & $\left[ 5.1065,7.182\right] $ & $2.07544$ & $\left[
0.7644,5.106\right] $ & $4.34129$ \\ \hline
CRI & $\left[ 6.3907,6.482\right] $ & $0.09119$ & $\left[ 1.4395,1.470\right]
$ & $0.03047$ \\ \hline
\end{tabular}

\begin{tabular}{ccccc}
\multicolumn{5}{c}{Table 6. continued$\ $} \\ \hline
Method & $\theta _{1}$ & Length & $\theta _{2}$ & Length \\ \hline
ACI & $\left[ 3.3657,17.9132\right] $ & $14.5475$ & $\left[ -5.708,13.635%
\right] $ & $19.3432$ \\ \hline
Boot\ -p CI & $\left[ 7.5863,15.3836\right] $ & $7.79726$ & $\left[
0.0014,5.8277\right] $ & $5.82632$ \\ \hline
Boot\ -t CI & $\left[ 8.7421,10.2884\right] $ & $1.5463$ & $\left[
3.2635,3.9634\right] $ & $0.69990$ \\ \hline
Boot-BC CI & $\left[ 7.1542,15.1540\right] $ & $7.99985$ & $\left[
0.0018,5.9530\right] $ & $5.95112$ \\ \hline
Boot-BCa CI & $\left[ 7.5863,15.3641\right] $ & $7.77776$ & $\left[
0.0015,5.8277\right] $ & $5.82622$ \\ \hline
CRI & $\left[ 10.4955,10.669\right] $ & $0.17349$ & $\left[ 3.7968,4.2318%
\right] $ & $0.43496$ \\ \hline
\end{tabular}

\begin{tabular}{ccccc}
\multicolumn{5}{c}{Table 6. continued$\ $} \\ \hline
Method & $\beta _{1}$ & Length & $\beta _{2}$ & Length \\ \hline
ACI & $\left[ -35.29,46.760\right] $ & $82.0509$ & $\left[ -0.061,0.121%
\right] $ & $0.18156$ \\ \hline
Boot\ -p CI & $\left[ 0.848,10.4994\right] $ & $9.65144$ & $\left[
0.0058,1.944\right] $ & $1.93823$ \\ \hline
Boot\ -t CI & $\left[ 0.3656,5.5876\right] $ & $5.22194$ & $\left[
-0.103,0.029\right] $ & $0.13181$ \\ \hline
Boot-BC CI & $\left[ 0.7926,9.7427\right] $ & $8.95014$ & $\left[
0.0091,2.455\right] $ & $2.44569$ \\ \hline
Boot-BCa CI & $\left[ 0.4894,9.1255\right] $ & $8.63611$ & $\left[
0.007,2.0050\right] $ & $1.99796$ \\ \hline
CRI & $\left[ 0.9100,4.2419\right] $ & $3.33191$ & $\left[ 0.006,0.0699%
\right] $ & $0.06389$ \\ \hline
\end{tabular}
\end{center}

After applying the approximation of Parsi and Bairamov \cite{35}, and to assess the effectiveness of that approximation, we get the expected values of the number of failures for the first production line which is $8.13025$, and the expected value of the number of failures for the second production line is 1.86975. While the exact number of failures after the test procedure is 7 for the first production line and 3 for the second production line.

\subsection{Real-data example}

In this subsection, the data of Xia et al. \cite{38} is used as an application of the JPROG-II-C model. This data represent the ordered breaking strengths of jute fiber at gauge lengths 10 mm and 20 mm. Jute Fiber has many uses in different sectors especially in the Engineering field such as the construction sector, automobile sector, textile sector, etc. The following references show the uses of Jute Fiber in different industries and its relation to the Engineering field: Shelar and Uttamchand \cite{39}, Zakaria et al. \cite{40}, Das et al. \cite{41} and Chakraborty et al. \cite{42}.

\begin{center}
\begin{tabular}{llllllllll}
\multicolumn{10}{c}{Table 7. Data Set 1} \\ \hline
$693.73$ & $704.66$ & $323.83$ & $778.17$ & $123.06$ & $637.66$ & $383.43$ & 
$151.48$ & $108.94$ & $50.16$ \\ \hline
$671.49$ & $183.16$ & $257.44$ & $727.23$ & $291.27$ & $101.15$ & $376.42$ & 
$163.40$ & $141.38$ & $700.74$ \\ \hline
$262.90$ & $353.24$ & $422.11$ & $43.93$ & $590.48$ & $212.13$ & $303.90$ & $%
506.60$ & $530.55$ & $177.25$ \\ \hline
\end{tabular}
\end{center}

The Kolmogorov Smirnov (K-S) distance between the empirical distribution of Data Set 1, shown in Table 7, and CDF of WGD is $0.105828$ with a $P$-value equals $0.855287$. Hence, the WGD fits well to Data Set 1.

\begin{center}
\begin{tabular}{llllllllll}
\multicolumn{10}{c}{Table 8. Data Set 2} \\ \hline
$71.46$ & $419.02$ & $284.64$ & $585.57$ & $456.60$ & $113.85$ & $187.85$ & $%
688.16$ & $662.66$ & $45.58$ \\ \hline
$578.62$ & $756.70$ & $594.29$ & $166.49$ & $99.72$ & $707.36$ & $765.14$ & $%
187.13$ & $145.96$ & $350.70$ \\ \hline
$547.44$ & $116.99$ & $375.81$ & $581.60$ & $119.86$ & $48.01$ & $200.16$ & $%
36.75$ & $244.53$ & $83.55$ \\ \hline
\end{tabular}
\end{center}

While the K-S distance between the empirical distribution of Data Set 2, shown in Table 8, and CDF of WGD is $0.149029$ with a $P$-value equals $0.473037$. Hence, the WGD also fits well to Data Set 2.

\begin{center}
\begin{tabular}{llllllllll}
\multicolumn{10}{c}{Table 9. JPROG-II-C data from the two data sets} \\ 
\hline
36.75 & 43.93 & 45.58 & 48.01 & 50.16 & 71.46 & 83.55 & 99.72 & 101.15 & 
108.94 \\ \hline
113.85 & 116.99 & 141.38 & 145.96 & 151.48 & 163.4 & 166.49 & 177.25 & 187.13
& 187.85 \\ \hline
\end{tabular}
\end{center}

Table 9 shows the JPROG-II-C data from the two Data Sets 1 and 2 with joint progressive censoring schemes $N=n_{1}+n_{2},~r=20$. $R=\left( 2,2,0,2,2,1,1,2,0,0,2,0,0,0,2,0,1,2,0,21\right)$ is chosen based on the choice of $s(1)$ and $s(2),$ where $s(1)=\left( 2,0,0,0,2,0,1,2,0,0,2,0,0,0,0,0,1,0,0,10\right) $ and $s(2)=\left( 0,2,0,2,0,1,0,0,0,0,0,0,0,0,2,0,0,2,0,11\right)$, such that $R_{i}=\sum_{h=1}^{2}s_{i}(h),~i=1,2,...,20$.

\begin{center}
\begin{tabular}{cccccccc}
\multicolumn{8}{c}{Table 10. Different point estimates for $\alpha
_{1},\theta _{1},\beta _{1};\alpha _{2},\theta _{2},\beta _{2}.$} \\ \hline
& $\left( .\right) _{ML}$ & $\left( .\right) _{Boot\;-p}$ & $\left( .\right)
_{Boot\;-t}$ & $\left( .\right) _{_{BS}}$ & \multicolumn{3}{c}{$\left(
.\right) _{BL}$} \\ \cline{6-8}
& \ \ \ \ \ \ \  & \ \ \ \ \ \ \ \ \ \ \ \  & \ \ \ \ \ \ \ \ \ \  & \  & $%
c=0.0001$ & $c=-2$ & $c=2$ \\ \hline
$\hat{\alpha}_{1}$ & $41.9214$ & $41.8926$ & $41.9102$ & $41.8855$ & $%
41.8855 $ & $41.8856$ & $41.8853$ \\ \hline
$\hat{\theta}_{1}$ & $16.3688$ & $16.3777$ & $16.4198$ & $16.4374$ & $%
16.4374 $ & $16.4407$ & $16.4343$ \\ \hline
$\hat{\beta}_{1}$ & $0.017$ & $0.0167$ & $0.0168$ & $0.017$ & $0.017$ & $%
0.017$ & $0.0169$ \\ \hline
$\hat{\alpha}_{2}$ & $55.6926$ & $55.6628$ & $55.6772$ & $55.6468$ & $%
55.6468 $ & $55.6485$ & $55.6451$ \\ \hline
$\hat{\theta}_{2}$ & $3.3669$ & $3.3663$ & $3.3658$ & $3.3655$ & $3.3655$ & $%
3.3655$ & $3.3655$ \\ \hline
$\hat{\beta}_{2}$ & $0.134$ & $0.1316$ & $0.1309$ & $0.1302$ & $0.1302$ & $%
0.1315$ & $0.129$ \\ \hline
\end{tabular}
\end{center}

The effectiveness of the proposed methods can be checked by comparing the results of these different methods. One can see that the results of Bayesian and maximum likelihood estimates are close together through the simulated and real examples. In addition, the authors showed the existence and uniqueness of the maximum likelihood estimates.

A graphical representations that show the existence and uniqueness of $\beta_1$ and $\beta_2$ based on Equation \eqref{3.2.4} are shown in Figure 3 (a) and (b) respectively.

\begin{figure}
     \centering
     \begin{subfigure}[b]{0.3\textwidth}
         \centering
         \includegraphics[width=\textwidth]{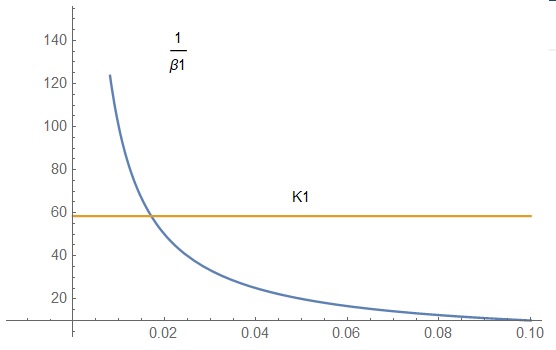}
         \caption{First production line}
         \label{fig:First production line}
     \end{subfigure}
     \hspace{1cm}
     \begin{subfigure}[b]{0.3\textwidth}
         \centering
         \includegraphics[width=\textwidth]{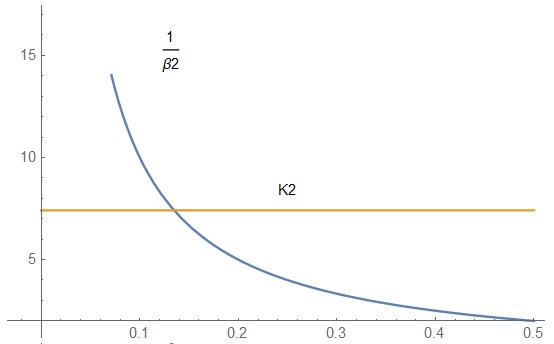}
         \caption{Second production line}
         \label{fig:Second production line}
     \end{subfigure}
        \caption{Graphical illustration for the existence and uniqueness of the MLEs in the real data example.}
        \label{fig:Figure 3}
\end{figure}

In the MCMC approach, we run the chain for $52000$ times and discard the first $2000$ values as `burn-in' while the prior knowledge parameters are chosen to be $\mu _{1}=$ $\mu _{2}=$ $q_{1}=q_{2}=p_{1}=p_{2}=0.01$ and $\lambda _{1}=$ $\lambda _{2}=w_{1}=w_{2}=v_{1}=v_{2}=2.$

\begin{center}
\begin{tabular}{ccccc}
\multicolumn{5}{c}{\textbf{TABLE 11} 95\% confidence intervals for $\alpha
_{1},\theta _{1},\beta _{1};\alpha _{2},\theta _{2},\beta _{2}.\smallskip $}
\\ \hline
\textbf{Method} & \textbf{$\alpha _{1}$} & \textbf{Length} & \textbf{$\alpha _{2}$} & \textbf{Length} \\ \hline
ACI & $\left[ 31.8799,51.9629\right] $ & $20.083$ & $\left[ 9.3128,102.073%
\right] $ & $92.7597$ \\ \hline
Boot\ -p CI & $\left[ 38.5958,46.8432\right] $ & $\allowbreak
8.\,\allowbreak 247\,4$ & $\left[ 44.3822,64.6178\right] $ & $20.2356$ \\ 
\hline
Boot\ -t CI & $\left[ 41.0657,41.5753\right] $ & $\allowbreak 0.509\,6$ & $%
\left[ 54.575,55.6084\right] $ & $1.0334$ \\ \hline
Boot-BC CI & $\left[ 38.8535,45.0456\right] $ & $\allowbreak 6.\,\allowbreak
192\,1$ & $\left[ 46.2642,61.6936\right] $ & $15.4294$ \\ \hline
Boot-BCa CI & $\left[ 39.1834,44.0334\right] $ & $\allowbreak
4.\,\allowbreak 85$ & $\left[ 48.3828,60.575\right] $ & $12.1922$ \\ \hline
CRI & $\left[ 41.867,41.9072\right] $ & $0.0401922$ & $\left[ 55.5918,55.7327%
\right] $ & $0.140813$ \\ \hline
\end{tabular}

\begin{tabular}{ccccc}
\multicolumn{5}{c}{Table 11. continued$\ $} \\ \hline
Method & $\theta _{1}$ & Length & $\theta _{2}$ & Length \\ \hline
ACI & $\left[ -22.0054,54.7431\right] $ & $76.7485$ & $\left[ 0.0693,6.6646%
\right] $ & $6.59529$ \\ \hline
Boot\ -p CI & $\left[ 3.10005,29.225\right] $ & $26.1249$ & $\left[
0.9597,5.3606\right] $ & $4.4009$ \\ \hline
Boot\ -t CI & $\left[ 15.4281,16.6201\right] $ & $1.192$ & $\left[
2.89935,3.79135\right] $ & $0.892$ \\ \hline
Boot-BC CI & $\left[ 6.0891,26.7269\right] $ & $20.6378$ & $\left[
1.60005,5.37795\right] $ & $3.7779$ \\ \hline
Boot-BCa CI & $\left[ 8.94385,23.2887\right] $ & $14.3448$ & $\left[
2.62805,4.22005\right] $ & $1.592$ \\ \hline
CRI & $\left[ 16.358,16.5554\right] $ & $0.197428$ & $\left[ 3.36,3.3718%
\right] $ & $0.0118217$ \\ \hline
\end{tabular}

\begin{tabular}{ccccc}
\multicolumn{5}{c}{Table 11. continued$\ $} \\ \hline
Method & $\beta _{1}$ & Length & $\beta _{2}$ & Length \\ \hline
ACI & $\left[ -0.0249,0.059\right] $ & $0.0838585$ & $\left[ -0.0848,0.3527%
\right] $ & $0.437472$ \\ \hline
Boot\ -p CI & $\left[ 0.0146,0.0658\right] $ & $0.0512$ & $\left[
0.0516,0.3503\right] $ & $0.2987$ \\ \hline
Boot\ -t CI & $\left[ 0.0018,0.024\right] $ & $0.0222$ & $\left[
0.2253,0.2438\right] $ & $0.0185$ \\ \hline
Boot-BC CI & $\left[ 0.0267,0.0345\right] $ & $0.0078$ & $\left[
0.2321,0.2634\right] $ & $0.0313$ \\ \hline
Boot-BCa CI & $\left[ 0.0209,0.0526\right] $ & $0.0317$ & $\left[
0.1145,0.1674\right] $ & $0.0529$ \\ \hline
CRI & $\left[ 0.0074,0.0298\right] $ & $0.02238$ & $\left[ 0.0694,0.2108%
\right] $ & $0.1414$ \\ \hline
\end{tabular}
\end{center}

The approximation of Parsi and Bairamov \cite{36} shows that the expected value of the number of failures for the first production line is $3.83187$ and the expected value of the number of failures for the second production line is $16.1681$. While the exact number of failures after the test procedure is $8$ for the first production line and 12 for the second production line.

\section{Conclusions}\label{sec7}

Using the JPROG-II-C samples strategy, the estimation of the $3k$ parameters of $k$ WG populations are performed based on Bayes and non-Bayes methods. Four types of bootstrap confidence intervals are used to obtain $95\%$ confidence intervals for the unknown parameters. The importance of the MCMC technique was noticeable in Bayesian estimation using the Metropolis-Hastings method. An illustrative example is presented to show how the MCMC and parametric bootstrap methods work. Also, the reliability of units that are manufactured by $k$ production lines is compared based on the invariance property of the MLEs of the parameters. The prior experimental knowledge of the number of failures for the first production line is computed approximately before the test procedure and compared with the exact number of failures after the test procedure. The proposed model can be applied to compare the reliability among different production lines through the failure times of the units produced by these production lines. These production lines may be located either at the same factory or in another place. The proposed model is recommended to be applied for evaluating the quality of the production lines and it can be a good reference for Engineering studies that are interested in the field of failure analysis. It also gives more reliable results whenever the size of selected samples for the study was relativity large. Since the proposed model depends on a parametric model, it is recommended to apply the it to real data sets that are characterized by high fitting criteria, measured basically through $p$-value and K-S distance for the WG distribution.

\nocite{*}
\bibliography{wileyNJD-AMA}%

\end{document}